\PassOptionsToPackage{final}{graphicx}
\PassOptionsToPackage{final}{hyperref}
\PassOptionsToPackage{notref,notcite}{showkeys}
\documentclass[final,acmsmall,screen,nonacm]{acmart}

\pdfoutput=1 %

\usepackage{ifthen}

\newboolean{spellingmode}
\setboolean{spellingmode}{false}

\usepackage{amsthm}
\usepackage{thmtools}
\usepackage{mathtools}
\usepackage{xspace}
\usepackage{booktabs}
\usepackage{subcaption}
\usepackage{stmaryrd} 
\usepackage{hyperref}
\usepackage{twshowkeys}
\usepackage[nothmenv]{twmath}
\usepackage{tikz}

\usetikzlibrary{backgrounds,calc,fit}

\usepackage{coqsrcref}

\usepackage{fixme}
\FXRegisterAuthor{tw}{atw}{TW}%
\FXRegisterAuthor{jm}{ajm}{JM}%

\newcommand{\nicehref}[2]{%
\begin{tikzpicture}[baseline=(txt.base)]
\node[text=blue!80!white!40!black,inner sep=0pt,text depth=1.1pt] (txt)
  {\href{#1}{{#2}{\hspace*{1pt}\raisebox{.0pt}{\color{blue!30!white}}}}};
\begin{scope}[on background layer]
\draw[color=blue!25!white] (txt.south west) -- (txt.south east);
\end{scope}
\end{tikzpicture}%
}

\newcommand{\onlineHtmlURL}{%
  https://arxiv.org/src/2602.16427v2/anc/html/%
}

\newsavebox{\logoagdabox}
\sbox{\logoagdabox}{%
  \begin{tikzpicture}[baseline={([yshift=2pt]logo.south)}]
  \node[inner sep=0pt] (logo) {\includegraphics[height=1em]{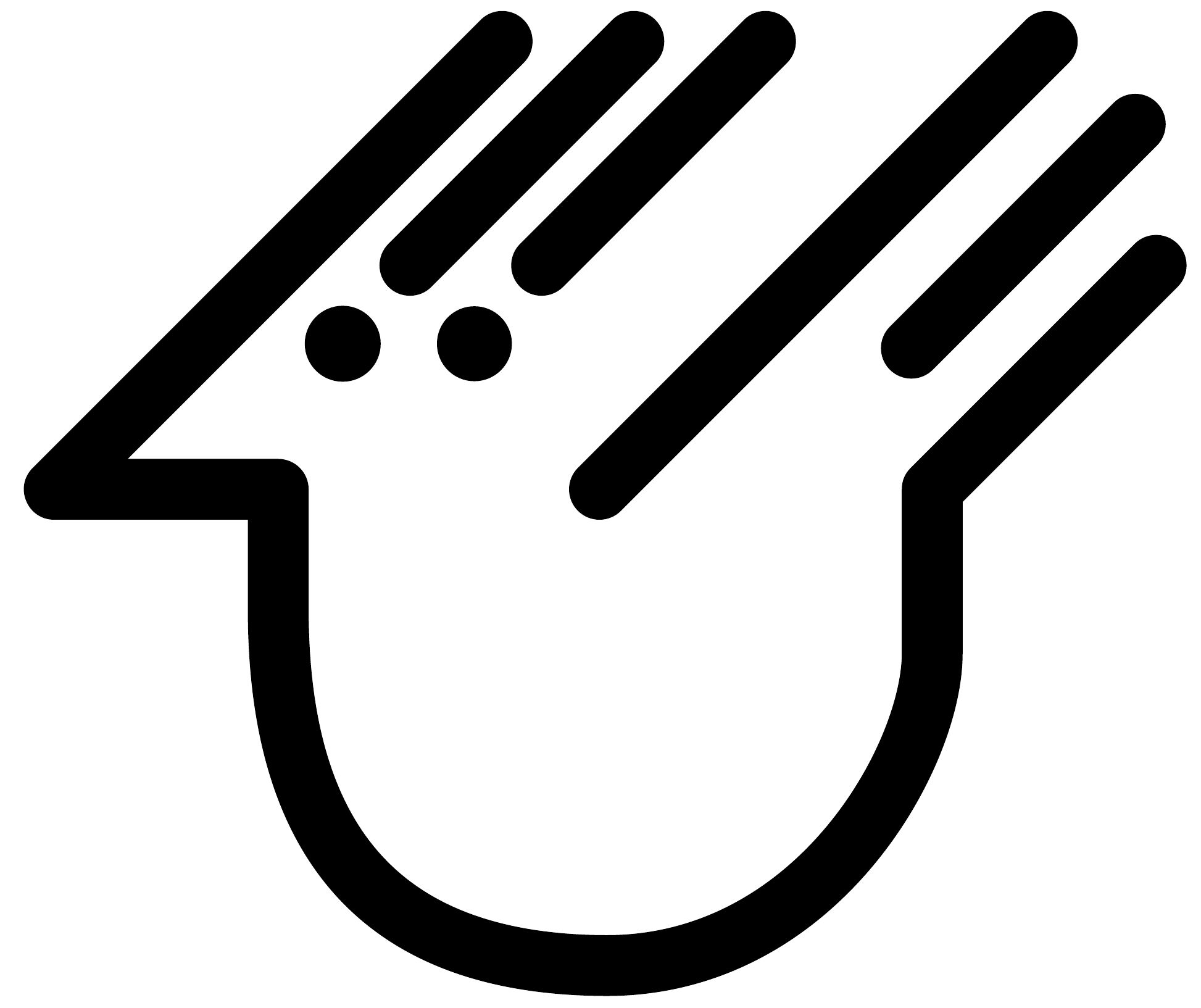}};
  \end{tikzpicture}%
}

\newcommand{\agdaref}[4][]{%
  \coqref{#2}{#3}{#4%
  }{\href{\onlineHtmlURL #2.html\##3}{\usebox{\logoagdabox}}}}
\newcommand{\agdarefcustom}[4]{%
  \coqrefcustom{#1}{#2}{#3}{\href{\onlineHtmlURL #2.html\##3}{\usebox{\logoagdabox}}}%
  {#4%
  }}

\ifthenelse{\boolean{spellingmode}}{%
	\AtEndPreamble{
	\hyphenpenalty=50000 %
	\pagestyle{empty}%
	\thispagestyle{empty}%
	\def\pagestyle#1{}%
	\def\thispagestyle#1{}%
	\def\labelmarginpar#1{}%
	}
}{}

\theoremstyle{plain}
\newtheorem{theorem}{Theorem}[section]

\newtheorem{definition}[theorem]{Definition}
\newtheorem{notation}[theorem]{Notation}
\newtheorem{observation}[theorem]{Observation}
\newtheorem{corollary}[theorem]{Corollary}
\newtheorem{remark}[theorem]{Remark}

\newcommand{\lstar}{\ensuremath{L^{*}}}
\newcommand{\lsharp}{\ensuremath{L^{\#}}}

\title{Formalized Run-Time Analysis of Active Learning -- Coalgebraically in Agda} %

\author{Thorsten Wißmann}
\orcid{0000-0001-8993-6486}
\affiliation{%
  \institution{Friedrich-Alexander-Universität Erlangen-Nürnberg}
  \city{Erlangen}
  \country{Germany}}


\newcommand{\D}{\ensuremath{\mathbb{D}}}
\newcommand{\Pow}{\ensuremath{\mathcal{P}}}

\newcommand{\DFA}{\ensuremath{\mathsf{DFA}}}
\newcommand{\Mealy}{\ensuremath{\mathsf{Mealy}}}
\newcommand{\MQ}{\ensuremath{\mathsf{MQ}}}
\newcommand{\EQ}{\ensuremath{\mathsf{EQ}}}
\newcommand{\CO}{\ensuremath{\mathcal{O}}}

\newcommand{\mathqt}[1]{\ensuremath{\mathord{\text{\textsf{\upshape `#1'}}}}}
\newcommand{\tick}{\ensuremath{\mathsf{tick}}}
\newcommand{\allows}{\ensuremath{\mathsf{allows}}}
\newcommand{\correct}{\mathqt{correct}\xspace}
\newcommand{\toohigh}{\mathqt{too-high}\xspace}
\newcommand{\toolow}{\mathqt{too-low}\xspace}
\newcommand{\halfopen}[1]{\ensuremath{[#1)}}
\newcommand{\floor}[1]{\ensuremath{\left\lfloor #1\right\rfloor}}
\newcommand{\ceil}[1]{\ensuremath{\left\lceil #1\right\rceil}}
\newcommand{\longto}{\longrightarrow}
\newcommand{\partialto}{\rightharpoonup}
\newcommand{\sem}[1]{\llbracket #1\rrbracket}
\newcommand{\modelsall}{\mathrel{\mathord{\,\models}\!\mathord{\forall\,}}}

\usepackage{newunicodechar}
\newunicodechar{⟦}{\ensuremath{\llbracket}}
\newunicodechar{⟧}{\ensuremath{\rrbracket}}
\newunicodechar{⊧}{\ensuremath{\mathord{\models}}}
\newunicodechar{∀}{\ensuremath{\mathord{\forall}}}
\newunicodechar{δ}{\ensuremath{\mathord{\delta}}}
\newunicodechar{∋}{\ensuremath{\mathord{\ni}}}
\newunicodechar{⁻}{\ensuremath{{}^{-}}}
\newunicodechar{¹}{\ensuremath{{}^{1}}}

\begin{document}
\begin{abstract}
  The objective of automata learning is to reconstruct the implementation of a hidden
  automaton, to which only a teacher has access. The learner
  can ask certain kinds of queries to the teacher to gain more knowledge about
  the hidden automaton. The run-time of such a learning algorithm is then measured in the
  number of queries it takes until the hidden automaton is successfully
  reconstructed, which is usually parametric in the number of states of that
  hidden automaton.
  How can we prove such a run-time complexity of learning algorithms
  in a proof assistant if we do not have the hidden automaton and the number of
  states available?

  In the present paper, we solve this by considering learning algorithms
  themselves as generalized automata, more specifically as coalgebras.
  We introduce formal and yet compact definitions of what a learner and a
  teacher are, which make it easy to prove upper and lower bounds of different
  kinds of learning games in the proof assistant Agda.

  As a running example, we discuss the common number guessing game where a teacher
  thinks of a natural number and answers guesses by the learner with
  `correct', `too high', or `too low'. To demonstrate our framework, we formally prove
  in Agda
  that binary search finds the teacher's secret number $n$ within $\mathcal{O}(\log n)$
  guesses and that no learning strategy can guarantee fewer than $\log_2(n)$ guesses.

  We apply our framework to automata learning to prove the following complexity
  bounds in Agda:
  (1) If the teacher does not provide counterexamples, then there is no learning
  algorithm that uses only polynomially many queries.
  (2) The $L^{\#}$-learning algorithm takes $\mathcal{O}(k\cdot n^2 + n\cdot \log(m))$
  queries to learn a Mealy machine with $k$ input symbols, $n$ states, where
  $m$ is the maximum length of the teacher's counterexamples.
\end{abstract}

\maketitle

\section{Introduction}
In her seminal paper~\cite{ANGLUIN198787}, Dana Angluin describes the problem
of reconstructing an unknown automaton-implementation from a black box by
considering a game between a learner and a minimally adequate teacher.
The teacher is in possession of a regular language, usually represented by a
deterministic finite automaton (DFA), and it is the task of the
learner to reconstruct
the automaton, by asking two kinds of questions to the teacher:
\begin{enumerate}
\item Membership queries: the learner can ask whether a word $w$ is accepted by
the automaton, to which the teacher answers with \textqt{yes} or \textqt{no}.
\item Equivalence queries: the learner can conjecture a concrete DFA and
present it to the teacher. If the automaton is correct, in the sense that the
DFA accepts the same language as the teacher's automaton, then the game ends and the learner has
succeeded. Otherwise, the teacher proves the conjecture wrong by returning a
\emph{counterexample}, that is, a word
for which the teacher's DFA and the conjectured DFA by the learner behave
differently.
\end{enumerate}
There may be multiple counterexamples, so the teacher enjoys a degree of
freedom when answering equivalence queries. Or in other words, there are
multiple minimally adequate teachers for the same hidden automaton.
On the learner side, many learning algorithms have
been developed in the past decades, including $\lstar$ from Angluin's
paper~\cite{ANGLUIN198787}, TTT~\cite{IsbernerHS14}, $\lsharp$~\cite{VaandragerGRW22},
and $L^\lambda$~\cite{HowarS22LLambda}. These algorithms differ in the
internal structure of their state space: some use tables to represent their
knowledge and some use trees. Moreover, from the teacher's point of view, they
differ in their \emph{strategy}, i.e.~in the number of queries they ask and
when they ask which query.
Generic learning algorithms~\cite{BarloccoEA19,ColcombetPS21,HeerdtS017,UrbatS20} have been
developed to unify learning algorithms for different automata types, with a
clear focus on the learning side of the game.
The study of lower bounds for learning, however, is concerned with the teacher's
side, because it entails the question of how \emph{adversarial}
a teacher can be~\cite{BalcazarDG97,KrugerGV23}.

In the present work, we lay foundations for
the formal verification of learning algorithms and the formalized study of lower
bounds of active learning by these
contributions:
\begin{enumerate}
\item We provide very concise definitions of the notions \emph{game type} (available queries),
  the queries' semantics, \emph{learner}, and \emph{teacher}. We cover both
  learning games for different kinds of automata and a natural number
  guessing game that serves as a running example.
\item We provide a notion of \emph{upper bound} of learning algorithms that
  does not rely on the teacher's hidden automaton but only the semantics of the
  teacher's responses. Our main \autoref{thm:learner-correct} is a generic
  proof principle to verify such an upper bound for learning algorithms.
  We demonstrate this principle for binary search and for the
  $\lsharp$~\cite{VaandragerGRW22} automata learning algorithm.
\item The above notions allow us to fully formalize our results in the proof
  assistant Agda, which marks a starting point for proving correctness and
  run-time complexity of learning algorithms in a formalized setting.
  Our formalization is available on
  \begin{center}
  \nicehref{\onlineHtmlURL index.html}{\onlineHtmlURL index.html}
  \end{center}
  Throughout the paper, we mark definitions and results with a clickable
  icon
  \href{\onlineHtmlURL index.html}{\usebox{\logoagdabox}}
  that points directly to the respective location in the formalization.
  For printed documents, \autoref{agdarefsection} lists this mapping
  explicitly and provides additional comments on formalization aspects.
  Concretely, we formalize the following results in Agda:
  \begin{enumerate}
  \item For natural number guessing,
  binary search takes at most $3 + 2\cdot\floor{\log_2(n)}$ queries (\autoref{exBinSearch}),
  where $n$ is the teacher's secret number. This bound is optimal up to a
  constant factor, because an adversarial teacher can force every learner to
  take more than $\floor{\log_2(m)}$ queries, even if learner and teacher agree
  on an interval of size $m$ in advance (\autoref{thmNatLower}).
  \item If in DFA-based automata learning, the equivalence query does not yield a counterexample but only a simple \textqt{yes}/\textqt{no}, then there is no polynomial learning algorithm (\autoref{thmNoPolyLearn}).
  \item The \lsharp{} learning algorithm for Mealy machines with $k$ input
  symbols takes at most 
    \[
        (k + 1) \cdot n \cdot (n + 1) + (n + 1) \cdot \ceil{\log_2 m} + 1
    \]
    queries to learn a Mealy machine with $n$ states, where $m$ is the maximum
    length of the teacher's counterexamples.
  \end{enumerate}
\end{enumerate}

\section{Key Ideas \& Natural Number Guessing}\label{sec:overview}
Before introducing formal definitions of automata learning, the present section explains the
conceptual problems that arise when formally proving the correctness of active
learning algorithms. We do so on a simple learning game on natural numbers that
exhibits every construct needed later for the formal analysis of automata
learning.

\begin{quote}\itshape
  Consider the learning game between a teacher and a learner, in which the teacher
  thinks of a natural number. It is the goal of the learner to reconstruct this number by
  consecutively guessing natural numbers.
  After each guess, the teacher replies with \correct, \toohigh, or
  \toolow.
  If the response is \correct, the learner wins and the game ends.
\end{quote}

How can we prove that winning strategies for the learner, that is
\emph{learning algorithms}, are correct?
By design of the game (both for natural number guessing and automata learning),
\emph{partial
correctness} comes for free: there is only one way for the learning game to
end, namely when the teacher confirms that the learner's hypothesis is correct.
Hence, the correctness of a learning algorithm boils down to proving a bound on
the number of queries until the teacher accepts the learner's hypothesis.

\begin{quote}\itshape
How can one prove a run-time bound that is parametric in an object which the
algorithm, and hence a formal proof, does not possess?
\end{quote}

Let us consider two examples of learning algorithms in the natural number
guessing game:

\begin{example}[Enumerator]\label{exEnumerator1}
Consider a learner that makes the guesses $0$, $1$, $2$, $3$, $\ldots$
enumerating all natural numbers. This learning strategy takes $1 + n$ guesses if
the teacher's secret number is $n$.
\end{example}
This bound is intuitively clear: if the teacher's secret is $0$, then the
learner needs one guess, if the teacher's secret is $1$, then it takes two
guesses and so on.
A formal proof of such a bound on guesses entails that a learning algorithm wins
against an arbitrary teacher. So for a rigorous formalization, one needs to
specify what learners, teachers, and rounds in the game are as mathematical
objects.

Intuitively, it is clear what a learning algorithm for natural number guessing
consists of:
\begin{enumerate}
\item The learner accumulates knowledge, so the learner may have a set $C$
  of internal states or memory, which is initially $q_0\in C$.
  In this memory, the learner can for example keep track of the teacher's
  responses.
  For example, for the enumerator (\autoref{exEnumerator1}), the internal state
  can simply be $C :=\N$, being the number which is guessed next.
\item The learner's next action consists of guessing a number in $\N$, which
  depends on the current internal state $q\in C$.
  If the guess is correct, then the game ends and so there are no obligations
  left for the learner.
  If the guess is wrong, there are two options: the guess was \toohigh{} or
  \toolow{}, and so the learner needs to specify how to proceed in either case.
\end{enumerate}

We turn this into a mathematical definition as follows:
\begin{definition}\label{natlearner}
  Let $W := \set{\toohigh,\toolow}$.
  A \emph{learner} consists of a set $C$ (called \emph{states}), a set $R$ (called \emph{results}), an initial state $q_0\in C$,
  and a map:
  \[
    c\colon C\longto \N \times R \times C^W
  \]
\end{definition}
Here, we write $C^W$ for the set of all maps $W\to C$.
The map $c$ describes the strategy of the learner: when in state $q\in C$, the only thing the learner can do
is to guess a number. The map $c(q) = (n,r,f)\in \N\times R \times C^W$ provides this guess $n\in \N$.
If the guess is correct, then the game ends immediately (with result $r\in R$,
which is discussed after \autoref{obsMoore}).  For the case in which the guess is
wrong, the learner provides a \textqt{continuation} function $f\colon W\to C$. Depending
on whether the guess $n$ was too high or too low, $f(w)\in C$ for $w\in
W=\set{\toohigh, \toolow}$ specifies in which state the learner is going to proceed.

Recall that a \emph{Moore automaton} $M = (Q,q_0,o,\delta)$ for an input alphabet $A$ and an output alphabet $O$
consists of a finite set of states $Q$, an initial state $q_0\in Q$, an
output function $o\colon Q\to O$, and a transition map $\delta\colon
Q\times A\to Q$.

\begin{observation}\label{obsMoore}
  A learner in the natural number guessing game is a (possibly infinite)
  Moore-automaton for the
  input alphabet $W$ and output alphabet $\N\times R$:
  \[
    q_0\in C
    \quad
    C\to (\N\times R)
    \quad
    C\times W \to C
  \]
  In state $q\in C$, such a Moore automaton has an output of type $\N\times R$
  and upon receiving an input $W$, the automaton transitions to a new state
  according to the map $C\times W\to C$.
\end{observation}
It may be surprising at first that the learner does not specify a successor
state for the case of a correct guess. But this is entirely analogous to a Moore
automaton (or a DFA), which does not specify a successor state for
\textqt{end of input}. Likewise, for learning, if the learning game ends, the
learner's strategy does not need to be evaluated any further.
If the learner has the desire
to pass additional information of some type $R$ (e.g.\ logs) to the outside world, then the
learner can do this by specifying respective $r\in R$ along with the query.
If not, the learner can put $R := 1$ (where $1$ is the one element set $1 :=
\set{0}$) and so $R$ vanishes
from the type $c\colon C\to \N\times C^W$.

\begin{example}[Binary search, \agdaref{NatExamples.SimpleLogLearner}{simple-bisect}{In Agda,
we use the full definition of learner that also contains a result set $R$.}]\label{exLogLearner}
  As the learner's state set, consider these intervals on $\N$:
  \[
    C := \set{\halfopen{n,\infty}\mid n\in \N} \cup \set{[n,m]\mid n, m\in \N}
    \cong \N + \N\times \N
    \qquad
    R := 1
    \qquad
    q_0 = \halfopen{0,\infty}
  \]
  where $+$ denotes the disjoint union of sets.
  The state of the learner describes the range in which the teacher's secret
  number must be.
  \[
    c(\halfopen{n,\infty}) = \big(2\cdot n+1, 
    \begin{array}{l@{\,}l@{}}
      \toolow &\mapsto \halfopen{2\cdot n+2,\infty} \\
      \toohigh &\mapsto [n, 2\cdot n]\\
    \end{array}\big)
    \in \N\times C^W
  \]
  \[
    c([n,m]) = \big(\floor{\frac{n+m}{2}}, 
    \begin{array}{l@{\,}l@{}}
      \toolow &\mapsto [\min(m,\floor{\frac{n+m}{2}}+1),\quad m] \\
      \toohigh &\mapsto [n, \quad\max(n, \floor{\frac{n+m}{2}}-1)]\\
    \end{array}\big)
    \in \N\times C^W
  \]
  Initially, the learner starts with half-open intervals and tries to find the
  next power of 2 that is bigger than the teacher's secret number. Once the
  learner has found out that the number is lower, the learner performs a binary search.
  Note that there is a difference between the learner \emph{already knowing} the 
  correct number and \emph{having guessed} it:
  if the learner ends up in the singleton interval $[n,n]$, then the learner knows
  the secret number must be $n$, but it takes one more guess to prove the learner right.
  A similar
  example run is listed in \autoref{run:knowB4guess}.
\end{example}
Dually to the learner, the teacher receives guesses and answers:

\begin{definition}[\agdaref{NatLearning}{NatTeacher}{}]\label{natteacher}
  A \emph{teacher} consists of
  a set $T$, an element $s_0\in T$, and a map:
  \[
    \delta\colon T \times \N \longto \set{\mathqt{correct}} + W\times T
  \]
  $T$ is the set of states and $s_0$ the initial state.
  We call a teacher \emph{stateless} if $|T|=1$.
\end{definition}
Intuitively, the teacher is in some state $s\in T$, initially $s_0$, which
represents what the teacher wants to remember between queries. In state $s\in T$, the response of the teacher 
to the guess $n\in \N$ is given by $\delta(s,n)$. If $\delta(s,n) = (w, s') \in W\times T$, then the 
guess $n$ was wrong: it was either $w=\toohigh$ or
$w=\toolow$. Furthermore, the teacher transfers to state $s'\in T$ for
the next guess by the learner. If however $\delta(s,n) = \mathqt{correct}$,
then the learning game ends, and so the teacher also does not need to specify a
successor state. This can be interpreted as either a correct guess by the
learner or by the teacher simply surrendering.
\begin{example}[\agdaref{NatExamples.HonestTeacher}{honest-teacher}{}]\label{exNatHonest}
  For every number $n\in \N$, there is a stateless teacher ($T:=1$) that
  answers queries in the expected honest manner.
  Note that we can also consider an honest teacher for a real number $r\in\R$:
  One can understand this as the task of learning the weight of an unknown
  item using a balanced scale and weights of 100g each. If the item weighs
  750g, then this may lead to a run as in \autoref{tab:contra}.
  Here, we interpret the learner as the winner because the teacher was convicted of inconsistency.
\end{example}
\begin{example}
Another example of a stateless teacher is the constant function that always
returns \toolow. An example run for this teacher is shown in
\autoref{tab:always2low}. Note that this run cannot be distinguished from an
honest teacher with the secret $10^{1000}+1$.
This is analogous to the application of automata learning to a
non-regular language: no matter how big the hypothesis DFA by the
learner is, the teacher will always find a longer counterexample.
\end{example}

\begin{example}[\agdaref{NatExamples.AdversarialTeacher}{adversarial-teacher}{In the formalization, the teacher's state is the set
  \[
    T = \set{ (c, [n,m]) \mid n\le m\in \N, c\in \N, 2^c \le 1 + (m-n) }
  \]
  Moreover, $\delta$ in the formalization explicitly distinguishes the cases whether
  the learner's guess is within or outside the teacher's interval.
  }]\label{exAdversarial}
  The \emph{adversarial teacher} has the set of non-empty intervals as its
  state space and splits this interval in half after each query:
  \[
    T = \set{ [n,m] \mid n \le m \in \N }
  \quad
    \delta([n,m],k) = \begin{cases}
      \correct &\text{if }n=k=m \\
      (\toolow,[\max(k+1,n),m]) & \text{if }k \le\floor{\frac{n+m}{2}} \\
      (\toohigh,[n,\min(k-1,m)]) & \text{if }k>\floor{\frac{n+m}{2}} \\
    \end{cases}
  \]
  Even though we have not yet defined the game and its semantics, one can see the idea that
  when initializing the teacher with the interval $s_0 = [n_0,m_0]$, then it takes
  any learner more than $\log_2(m_0-n_0)$ queries to obtain the affirmative
  $\correct$ from the adversarial teacher. Later in \autoref{thmNatLower}, we
  will make this observation formal.
\end{example}
\begin{table}
  \begin{subfigure}{.3\textwidth}
    \begin{tabular}{@{}ll@{}}
      \toprule
      Query
      & Response \\
      \midrule
      1 & \toolow \\
      3 & \toohigh \\
      2 & \correct \\
      \bottomrule
      \end{tabular}
      \caption{Learner wins}
      \plainlabel{run:knowB4guess}
  \end{subfigure}
  \begin{subfigure}{.3\textwidth}
    \begin{tabular}{@{}ll@{}}
      \toprule
      Query
      & Response \\
      \midrule
      $5$ & \toolow \\
      $8$ & \toohigh \\
      $7$ & \toolow \\
      \bottomrule
      \end{tabular}
      \caption{Learner wins}
      \plainlabel{tab:contra}
  \end{subfigure}
  \begin{subfigure}{.3\textwidth}
    \begin{tabular}{@{}ll@{}}
      \toprule
      Query
      & Response \\
      \midrule
      10 & \toolow \\
      $10^{10}$ & \toolow \\
      $10^{1000}$ & \toolow \\
      \bottomrule
      \end{tabular}
      \caption{Teacher wins}
      \plainlabel{tab:always2low}
  \end{subfigure}
  \caption{Example runs of the natural number guessing game, when played for 3 queries}
\end{table}

\section{Generic Types of Games, Learners, and Teachers}
\begin{notation}
We write $1$ and $2$ for the sets of the respective cardinality:
$1=\set{0}$, $2=\set{0,1}$.
For sets $X$ and $Y$, we write $Y^X$ for the set of functions $X\to Y$.
In particular, we write $2^X$ for the set of predicates on $X$, which we
implicitly identify with subsets $2^X \cong \Pow(X)$ whenever it simplifies the
notation. For a subset $S\subseteq X$, we write $S\hookrightarrow X$ to denote
the canonical inclusion function.

We write $X+Y$ for the disjoint union of sets $X$ and $Y$.
The projection
functions for cartesian products $X\times Y$ are called $\prl\colon X\times
Y\to X$ and $\prr\colon X\times Y\to Y$.
\end{notation}

We now generalize the ideas from the previous section to make them applicable
to different types of learning games, starting with the definition of what a
\emph{type of game} should be:
\begin{definition}[\agdaref{Game-Basics}{Game-Type}{We do not need
  functoriality in $R$ in the formalization. Also, we define
  the action of $F$ on maps via \texttt{F-map} as part of the game semantics
  structure.}]
  A \emph{game type} is a functor $F\colon \Set\times \Set\to \Set$.
  Concretely, a game type $F$ is a construction that sends sets $R$ and $X$ to a set $F(R,X)$.
\end{definition}
\begin{remark}
  Functors are a standard concept from category theory~\cite{awodey2010category}.
  For the sake of the present work, it suffices to understand a functor as a
  construction that takes two sets $R$ and $X$ and yields a new set $F(R,X)$.
  One of the axioms of functors states that the construction is compatible with
  maps $f\colon X\to Y$ in the sense that they induce a map $F(R,f)\colon
  F(R,X)\to F(R,Y)$.
\end{remark}
\begin{example}[\agdaref{NatLearning}{F}{}]
  The natural number guessing game has the game type
  \[
    F(R,X) = \N \times R \times X^W.
  \]
\end{example}
\begin{definition}[\agdaref{Game-Definitions}{Learner}{}]
  A \emph{learner} (for game type $F$) consists of a set of states $C$, a set
  of results $R$, an initial state $q_0\in C$ and a map:
  \[
    c\colon C\longrightarrow F(R,C)
  \]
\end{definition}
\begin{remark}
  A learner is a pointed coalgebra for $F(R,-)$~\cite[Def.~3.3]{wellPointedCoalg}.
\end{remark}
If we instantiate $F$ to be the number guessing game type, the learner materializes
to what we have defined before in explicit terms (\autoref{natlearner}).
\begin{example}[\agdaref{DFALearning}{F}{}]
  For automata learning, fix an input alphabet $A$ and let $\DFA$ be the set of
  all deterministic finite automata over $A$. We put as the game type:
  \[
    F(R,X) := A^*\times X^2 + \DFA \times R \times X^{A^*}
  \]
  The first summand (which we call $\MQ$) describes membership queries and the second summand
  (called $\EQ$) describes equivalence queries. For a learner with states $C$, a map $c\colon
  C\to F(R,C)$ means that in each state $q\in C$ the learner can choose between
  two types of queries:
  \begin{enumerate}
  \item In order to perform a membership query for a word $w\in A^*$, the learner
    sets $c(q) = \MQ(w, f)$ where $f\colon 2\to C$ is the continuation: if $w$
    was in the teacher's language, the learner continues in $f(1)$ and
    otherwise, the learner will continue in state $f(0)$.
  \item If the learner conjectures in state $q$ that the hidden automaton is $H\in \DFA$,
    then the learner specifies $c(q) = \EQ(H,r,f)$ where $r$ is additional
    information (can be ignored by $R=1$) and $f\colon A^*\to C$ specifies how
    the learner continues. If $w\in A^*$ is a counterexample to the hypothesis
    $H$, then the learner proceeds in state $f(w)\in C$.
  \end{enumerate}
  Note that the summand $A^*\times X^2$ for membership queries did not mention the
  result type $R$: this is because the learning game is not ended by membership
  queries. So the notion of game type does not only encode the type of queries
  but also which queries may end the learning game and which may not (cf.\
  \autoref{dfateacher} for the construction). This is achieved by
  the following generic definition of what a teacher is:
\end{example}

\begin{definition}[\agdaref{Game-Definitions}{TeacherNaturality.Teacher}{}]\label{defTeacher}
  A \emph{teacher} (for game type $F$) consists of a set of states $T$, an
  initial state $s_0\in T$ and a natural family of maps
  \[
    \alpha_{R,X}\colon
    F(R,X)
    \longrightarrow
    (R + X \times T)^T
    \qquad
    \text{for all sets $R$ and $X$}
  \]
  In other words, let $M_T\colon \Set\times\Set\to\Set$ be the functor
  \(
    M_T(R,X) = (R + X\times T)^T
  \).
  Then, a teacher consists of the following data:
  \[
    \text{a set }T, \qquad s_0\in T,\qquad
    \text{ and }
    \qquad
    \text{ a natural transformation }~
    \alpha\colon 
    F\longto M_T.
  \]
\end{definition}
\begin{remark}
  Naturality~(\agdarefcustom{Naturality}{Game-Definitions}{TeacherNaturality}{})
  enforces that $\alpha$ cannot introspect the types $R$ and $X$;
  we refer to~\cite{awodey2010category} for the categorical details which are not relevant
  for the present work.
\end{remark}
\begin{example}[\agdaref{NatLearning}{mkTeacher}{Its inverse is defined in the same file.}]
  For $F(R,X) = \N \times R \times X^W$, teachers in the generic sense
  $\alpha\colon F\to M_T$
  for a set $T$ correspond to the maps as in \autoref{natteacher} before:
  \[
    \begin{array}[t]{cc}
      \alpha_{R,X}\colon \N\times R \times X^W
      \longto
      (R + X \times T)^T
      \\
      \text{Natural in }R\text{ and }X
    \end{array}
    \quad
    \Longleftrightarrow
    \quad
    \delta\colon
      T\times \N\longto \set{\correct} + W \times T
  \]
  The correspondence is essentially the Yoneda lemma~\cite{awodey2010category}:
  Given a natural transformation, we instantiate $R := \set{\correct}$ and $X:=W$ and then put
  \[
    \delta(t,n) = \alpha_{\set{\correct},W}(n,\correct,\id_W)(t)
    \tag*{\agdarefcustom{\textbullet\ \ensuremath{\delta}}{NatLearning}{mkTeacher-on⁻¹}{}}
  \]
  Conversely, given $\delta$, it is not hard to verify that the obvious
  candidate for $\alpha_{R,X}$ is indeed natural in $R$ and $X$.
\end{example}
\begin{example}[\agdaref{DFALearning}{DFATeacher}{Conversion functions back and forth to the generic teacher notion are in the same file.}]\label{dfateacher}
  For the game type $F(R,X) := A^*\times X^2 + \DFA \times R \times X^{A^*}$ of
  DFA-Learning, our notion of teacher corresponds to a pair of functions, one
  for membership and one for equivalence queries:
  \[
    \begin{array}[t]{cc}
    A^*\times X^2 + \DFA \times R \times X^{A^*}
    \xrightarrow{\alpha_{R,X}}
    (R + X \times T)^T
    \\
    \text{Natural in }R\text{ and }X
    \end{array}
    ~\Longleftrightarrow~
    \left\{
      \begin{array}{@{}r@{\,}l@{}}
      \MQ\colon& T \times A^* \longto 2 \times T \\
      \EQ\colon&  T \times \DFA \longto 1+ A^* \times T \\
      \end{array}
    \right.
  \]
  For a teacher's state $s\in T$, the query $\MQ(s,w)$ returns whether
  $w$ is accepted by the teacher's automaton, together with the teacher's new state.
  The equivalence query $\EQ$ takes a teacher state $s\in T$ and a hypothesis $H\in
  \DFA$. If the hypothesis is correct, then $\EQ(s,H) \in 1$. Otherwise $\EQ$
  returns a counterexample $w\in A^*$ and a successor state.
  The conversion from right to left is straightforward. For the direction from
  left to right, we instantiate $R := 1 = \set{0}$ and $X := A^*$ in order to obtain
  $\EQ(s,H) := \alpha_{1,A^*}((H,0, \id_{A^*}))(s)$. For the membership query,
  we however instantiate $R:=\emptyset$ and $X:=2$, because this forces
  $\alpha_{\emptyset,2}$ to return something in the right-hand component of the
  disjoint union: $\MQ(s,w) = \alpha_{\emptyset,2}(w,\id_2)(s)$.
\end{example}
So the generic notion of teacher instantiates suitably for different game
types. Moreover, by instantiating $X := C$ for a learner on $C$, we can compose
$\alpha$ directly with a learner to define the actual game that is performed
when letting learner and teacher interact: 
\begin{definition}[\agdaref{Game-Definitions}{Game.game-step}{}]
  For a learner $(C,c,q_0)$ and a teacher $(T,\alpha,s_0)$ (of the same game
  type $F$), a \emph{game} is a possibly infinite sequence of state pairs $(q_k,s_k) \in C\times
  T$, $k\in \N$, that starts with the players' initial states $(q_0,s_0)$ and is induced by the map
  \[
    g(q,s) = \alpha_{R,C}(c(q))(s)
    \qquad
    g\colon
    C\times T \longto
    R+ C\times T
  \]
  If $(q_k,s_k)$ is defined and $g(q_k,s_k) \in C\times T$, then $(q_{k+1},s_{k+1}) := g(q_k,s_k)$.
\end{definition}

\subsection{Game Semantics}
For the semantics of a game, we
fix a set $\D$, which is intuitively the domain from which the teacher picks a
secret element. In the learning literature~\cite{Angluin90neg}, the set $\D$ is
typically called the \emph{class of concepts}. For natural number guessing, we
have $\D = \N$ and in the case of DFA learning we have $\D=\DFA$.

\begin{definition}[\agdaref{Game-Definitions}{GameSemantics}{}]
  The \emph{semantics} of a game type $F$ consists of
  a family of maps $\sem{-}_X\colon F(R,X)\to F(R,2^\D\times X)$
  such that $F(R,\prr) \cdot \sem{-}_X = \id_{F(R,X)}$.
\end{definition}
Intuitively, the map $\sem{-}$ equips each $x\in X$ that is mentioned in $y\in
F(R,X)$ with a predicate $P\in 2^\D$ describing the information when the
teacher tells to proceed in $x$.
\begin{example}[\agdaref{NatLearning}{NatSem.⟦F\_⟧}{}]
  For natural number guessing $F(R,X) = \N \times R \times X^W$, we have:
  \[
    \sem{(H,r,f)} = \big(H,r,
      \begin{array}{@{}l@{\,}l@{}}
        \toolow &\mapsto (\set{d\in \N\mid H < d}, f(\toolow)) \\
        \toohigh &\mapsto (\set{d\in \N\mid H > d}, f(\toohigh)) \\
      \end{array}
    \big)
  \]
  So the semantics specifies that if the successor state $f(\toolow)$ is taken,
  then the secret number $d$ is greater than the hypothesis $H < d$. Analogously,
  continuing in $f(\toohigh)$ means that the teacher's secret number $d$ satisfies $H>d$.
\end{example}
\begin{example}[\agdaref{DFALearning}{DFASem.⟦F\_⟧}{}]
  For DFA learning $F(R,X) = A^*\times X^2 + \DFA \times R \times X^{A^*}$ and $\D=\DFA$, we define:
  \begin{align*}
    \sem{\MQ(w,f)} &= \MQ(w,
        b \mapsto (\set{M\in \DFA\mid L(M)(w) = b},~ f(b))
        )
    \\
    \sem{\EQ(H,r,f)} &= \EQ(H,r,
        w \mapsto (\set{M \in \DFA \mid L(M)(w) \neq L(H)(w)},~ f(w))
        )
  \end{align*}
  The semantics of a membership query $\MQ(w,f)$ is that in the successor state $f(b)$ we now learn that the teacher's DFA
  $M$ satisfies $L(M)(w) = b$, where $L(M)\colon A^*\to 2$ is the language accepted by $M$.
  In the equivalence query $\EQ(H,r,f)$, any word $w\in A^*$ returned by the teacher means that for $w$, the teacher's DFA $M$ produces a different output than $H$, i.e.: $L(M)(w) \neq L(H)(w)$.
\end{example}
We can now instantiate the teacher's natural transformation $\alpha_{R,X}$ to
$X := 2^\D \times C$ (for some learner $(C,c,q_0)$) to obtain the semantics of
the teachers in state $s\in T$ to the learner's query in state $q\in C$:
\[
  c(q) \in F(R,C)
  ~~\Rightarrow~~
  \sem{c(q)} \in F(R,2^\D\times C)
  ~~\Rightarrow~~
  \alpha(\sem{c(q)})(s) \in R + 2^\D\times C\times T
\]
\begin{definition}[\agdaref{Game-Definitions}{Game.still-possible}{}]
  We say that $d\in \D$ is \emph{(still) possible} 
  after $n\in \N$ rounds starting from 
  states $(q,s)\in C\times T$ if $n=0$ or if all of these conditions hold:
  \begin{enumerate}
  \item $n>0$
  \item $\alpha(\sem{c(q)})(s) = (P,q',s') \in 2^\D\times C\times T$ ~(the game does not end in this round)
  \item $d\in P$ ~(the next response is consistent with $d$)
  \item $d$ is still possible after $n-1$ rounds starting from $(q',s')$.
  \end{enumerate}
\end{definition}

If $d\in \D$ is still possible after $n$ rounds, it means that judging from the
teacher's $n$ responses, we do not yet know whether $d$ might be the teacher's
secret. Note that this in particular entails that the teacher has not accepted
any of the learner's hypotheses. This is because if the teacher accepts a
learner's hypothesis, then the learner knows for every $d\in \D$ whether $d$ was the
teacher's secret or not.
\begin{definition}[\agdaref{Game-Definitions}{Game.learner\_guesses\_within\_rounds}{}]
  Complementarily, we say that starting from $(q,s)\in C\times T$, the learner
  finds $d\in \D$ within $n\in \N$ rounds, if $d$ is not possible after $n$
  rounds.
\end{definition}
Intuitively, finding $d\in \D$ within $n\in \N$ rounds covers two cases:
\begin{enumerate}
\item If $d$ happens to be the teacher's secret, the learner makes the teacher
accept an equivalence query within $n$ rounds.
\item If $d$ is not the teacher's secret, then the next $n$ responses by
the teacher will refute $d$.
\end{enumerate}
So $n$ can be considered as the maximum bound of queries that a learner needs to learn $d$:
\begin{example}[\agdaref{NatExamples.LinearLearner}{linear-learner-bound}{}]
  In the natural number game, the learner defined by
  \[
    C := \N
    \quad
    R := 1 = \set{*}
    \quad
    q_0 = 0
    \quad
    c(n) = (n,*, w\mapsto n+1)
  \]
  finds every $d\in \N$ within $1+d$ rounds for every teacher.
\end{example}

\subsection{Proving Learner Run-Time}
We provide a general recipe for proving such upper bounds on
queries for generic game types $F$. To this end, we model bounds by a bound
function $b\colon \D\to \N$
denoting for each concept $d\in \D$ the maximum number of queries needed to
find $d$. We use the standard technique of a stepping function, but intertwine
it with game semantics. For this, we consider the following predicate lifting.
\begin{definition}[Predicate Lifting, \agdaref{Game-Definitions}{GameSemantics.\_⊧F∀\_}{
  To make $\modelsall$ easier to work with in the implementation, we require
  the semantics of $F$ to implement $\modelsall$ explicitly as a predicate
  describing those $Z\in FX$ that are in the image of
  $F(R,P\hookrightarrow X)\colon F(R,P)\to F(R,X)$.
}]
  Given sets $R,X$, an element $Z\in F(R,X)$, and a predicate $P\subseteq X$, we define the relation $\modelsall$ by
  \[
    Z\modelsall P
    \quad:\Longleftrightarrow\quad
    \text{there is }Z'\in F(R,P)
    \text{ with }F(R, P\hookrightarrow X)(Z') = Z
  \]
\end{definition}
Intuitively, $\modelsall$ describes a $\Box$-modality: it specifies that all
successor states $x\in X$ that are buried in the structure $Z\in F(R,X)$ have
the property $P$.
We use this modality to provide sufficient conditions for the learner to win
the generic learning game with a bounded number of queries:

\begin{definition}[\agdaref{Learner-Correctness}{Learner-Stepping}{}]\label{defStepBounded}
  A learner $(C,c,q_0)$ for game type $F$ is \emph{step-bounded} by a function $b\colon \D\to \N$
  if the learner can be equipped with the following data:
  \begin{enumerate}
  \item a function $\tick\colon C\to \N$, describing the number of
  queries performed so far.
  \item a function $\allows\colon C\to 2^\D$, describing the concepts from $\D$
  that have not been refuted by the teacher so far.
  \item $\allows(q_0) = \D$ (the initial state does not exclude any concept)
  \item whenever $d\in \allows(q)$, then $\tick(q) < b(d)$ for all $d\in \D$, $q\in C$.
  \item the transition structure $c\colon C\to F(R,C)$ satisfies for all $d\in \allows(q)$:
  \begin{equation}
    \sem{c(q)} \modelsall \set{ (K,q') \in 2^\D\times C
      \mid
      \text{if }d \in K\text{ then }d\in\allows(q')\text{ and }\tick(q) <\tick(q')
    }
  \end{equation}
  \end{enumerate}
\end{definition}
The last item models the central criterion for correctness, namely that
$\allows$ is preserved by $c$. Here, we
have $\sem{c(q)} \in F(R,2^\D\times C)$, so we can apply
$\modelsall$ to a predicate on $2^\D\times C$. The condition means that the learner 
may update their own knowledge ($\allows$) according to the teacher's response:
assume the learner is in state $q\in C$ which allows $d\in \D$ and the teacher
gives a response which (a) is consistent with the set $K\subseteq \D$ of concepts and (b) which makes
the learner transition to state $q'$. Then this successor state $q'$ must still allow $d$ 
and must have an advanced step counter.
Note that the condition only models an inclusion: the learner is still allowed
to \emph{forget} previous information by increasing the $\allows$-predicate.

In the example of a query $c(q) = (H,r,f) \in \N\times R\times C^W$ in the number
guessing game, the preservation criterion instantiates to:
\[
  \text{for all }d\in \allows(q):
  \begin{array}{@{}l@{\,}l}
  \text{ if }d < H\text{ then }d \in \allows(f(\toohigh))&\text{ and }\tick(q) < \tick(f(\toohigh)) \\
  \text{ if }d > H\text{ then }d \in \allows(f(\toolow))&\text{ and }\tick(q) < \tick(f(\toolow))
  \end{array}
\]
In particular, if the learner knows that certain responses by the teacher are
inconsistent with what the learner has already inferred, then the condition
becomes vacuous and the learner may transition to an arbitrary
state:
\begin{itemize}
\item In the number guessing game, an example for such a response is \toohigh
when the learner was guessing \textqt{0}.

\item In DFA learning, Item 5 ensures that if the learner's hypothesis
automaton is consistent with all previous membership queries, then the
counterexample produced by an equivalence query must be different from all the
membership queried words before.
\end{itemize}
\begin{example}[\agdaref{DFA-Examples.Enumerator}{enumerator-correct.enumerator-step-bounded}{}]
  In DFA-Learning, let $e\colon \N\to \DFA$ and $i\colon \DFA\to \N$ be functions
  such that for every $M\in\DFA$, the automata $M$ and $e(i(M))$ accept the same
  language. Then the learner that asks equivalence queries $e(0), e(1), e(2),
  \ldots$ is step-bounded by $b(M) := 1+i(M)$.
\end{example}

\begin{theorem}[\agdaref{Learner-Correctness}{Correctness.learner-correct}{Proof details are in the module above}]\label{thm:learner-correct}
  If a learner is step-bounded by $b\colon \D\to \N$, then for every teacher
  and every $d\in \D$, the learner finds $d$ within $b(d)$ rounds.
\end{theorem}

\section{Case Study: Natural Number Guessing}
We instantiate the above generic notions to obtain upper and lower bounds of
the natural number guessing game in a formalized setting:

\begin{example}[\agdaref{NatExamples.LogLearner}{log-learner-bound}{}]\label{exBinSearch}
  In natural number guessing, we refine the previous binary search
  learner (\autoref{exLogLearner}) as follows to prove that the binary search
  is step-bounded by $b(d) = 3 + 2\cdot \floor{\log_2(d)}$.
  To avoid an overly syntactic description, we identify the states with their
  interpretation as intervals and define:
  \[
    C := \set{\halfopen{0,\infty}}
    \cup \set{\halfopen{2^b,\infty}\mid b\in \N}
    \cup \set{\halfopen{b, b+2^e}\mid b\in \N, e\in \N, e\le \log_2(b)}
    \tag*{\agdarefcustom{\textbullet\ State space}{NatExamples.LogLearner}{LogLearner}{}}
  \]
  In the third disjunct, we phrase the inequality using $\log_2$ because this
  also covers the case $b=0$ (in Agda, $\log_2(0) = 0$ by definition). By this
  inequality, the learner knows in each state how many queries have been asked already:
  \[
    \tick(\halfopen{0,\infty}) = 0
    \quad
    \tick(\halfopen{2^b,\infty}) = 1 + b
    \quad
    \tick(\halfopen{b,b+2^e}) = 2 + 2 \cdot \floor{\log_2(b)} - e
    \tag*{\agdarefcustom{\textbullet\ \ensuremath{\tick} function}{NatExamples.LogLearner}{ticks}{}}
  \]
  If a state $q\in C$ allows $d\in \N$ (i.e.~$d\in \allows(q)$), then
  $\tick(q) < b(d)$ (\agdarefcustom{\textbullet\ \ensuremath{\tick <
  b}}{NatExamples.LogLearner}{in-time}{}).
  The main effort is then to prove that the map $c$ is compatible with $\tick$
  and $\allows$ (Item 5) \agdarefcustom{\textbullet\ $c$ compatibility}{NatExamples.LogLearner}{δ-preserves-∋}{}. Then one obtains:
  the binary search learner finds every $d\in \N$ within $3+2\cdot
  \floor{\log_2(d)}$ rounds.
\end{example}

Note that this upper bound on queries holds for all teachers, including the
adversarial teacher (\autoref{exAdversarial}) that made use of its internal
state space. How can we prove that binary search is indeed the fastest learning
algorithm?
When proving such a lower bound in a formalized setting, one needs to be
careful to put the quantifiers in the right order, because otherwise, one
obtains statements like the following:
\begin{proposition}[\agdaref{NatExamples.GuessingLearner}{theorems.learner-finds-every-secret}{}]\label{win2rounds}
  Fix any stateless teacher 
  that answers something different from $\toolow$ for at least one query.
  Then there is a learner that finds every secret $d\in \N$ within $2$ rounds.
\end{proposition}
Despite being counter-intuitive, the statement holds because fixing the teacher upfront
means that among all the many learners there is one that just luckily guesses the
teacher's secret directly on the first try. The bound in the proposition is 2
because it may need a second query to convict the teacher of inconsistency (cf.\
\autoref{tab:contra}).

\begin{theorem}[Lower bound, \agdaref{NatExamples.AdversarialTeacher}{thm-adversarial-teacher}{The Agda standard library has the convention $\log_2(0)=0$}]\label{thmNatLower}
  In the natural number guessing game, if the adversarial teacher
  (\autoref{exAdversarial}) is initialized with an interval $[n,n+m]$ (for $n,m\in \N$),
  then for every learner, there is a natural number $d\in \N$ such that $d$ is still
  possible after $\floor{\log_2(m)}$ rounds.
\end{theorem}
So even if learner and teacher agree on a finite interval of size $m$ in
advance, the teacher can still force the learning game to take $\log_2(m)$ steps.
\begin{example}[\agdaref{NatExamples.AdversarialTeacher}{example-adversarial-100}{}]
  In a television show in 2016, Steve Ballmer of Microsoft presents a job interview
  question: 
  \begin{quote}\itshape 
    I'm thinking of a number between 1 and 100. You can guess. After each
    guess, I will tell you whether you're \textqt{high} or \textqt{low}.
    If you get it the first guess, I give you 5 bucks, [otherwise], 4 bucks, 3,
    2, 1, 0, you pay me a buck, you pay me 2, you pay me 3. The question is: do
    you want to play or not?
    \hfill{\upshape \url{https://youtu.be/svCYbkS0Sjk}}
  \end{quote}
  In the video, the TV show host needs 7 guesses, and indeed by
  $\floor{\log_2(m)}=6$, \autoref{thmNatLower} shows that it is not
  possible to win the game in fewer than 7 guesses if we have no guarantee that
  the teacher has changed their mind (i.e.~updated their internal state) during
  the game.
  If on the other hand, the teacher chooses a number uniformly random and sticks
  to it, the expected revenue is positive at $\$0.20$~(cf.~John
  Graham-Cumming's blog post~\cite{jgc}).
\end{example}

\section{Case Study: Automata Learning}

Often, the notion of the automata learning game needs to be
adjusted slightly to new surrounding conditions, e.g.\ to different models like
Mealy and Moore automata and to different kinds of queries. These adjustments
can easily be covered by the above general definitions because it suffices to
define an adjusted game type and its semantics.

\subsection{Size of Counterexamples}
\label{secCeSize}
If we do not assume that the teacher provides counterexamples of shortest possible length,
then the run-time of common learning algorithms also needs to be parametric in
the length of the longest counterexample by the teacher.
However, this maximum length must not be passed to the learning algorithm as a
parameter, since it would simplify the learning problem.
Instead,
we incorporate
the length by $\D := \DFA\times \N$ and adjust the semantics of queries as:
\[
  \sem{\EQ(H,r,f)} = \EQ(H,r,w\mapsto (\set{(M,\ell)\mid L(M)(w) \neq L(H)(w)\text{ and } |w|\le \ell},~ f(w)))
  \tag*{\agdarefcustom{Counterexample Size}{DFA-CE-Size-Learning}{DFA-CE-Sem.⟦F\_⟧}{}}
\]
Then, the bound $\D\to \N$ can make use both of the size of DFAs and the maximum counterexample length, e.g.\ when verifying the run-time of the binary-search based counterexample processing by Rivest \&
Schapire \cite{RivestS89,RivestS93}. At the same time, the learner does not
know this maximum length during the learning game.

\subsection{Lack of Counterexamples}
Consider the automata learning game in which the teacher does not provide a
concrete counterexample on wrong equivalence queries but simply a \textqt{no}.
In this so-called \emph{restricted} automata learning~\cite[Sec.\
3.2]{Angluin87ML} one can show that there is no learning algorithm that uses
only a polynomial number of queries.

The game type is:
\[
  F(R,X) := A^*\times X^2 + \DFA \times R \times X
\]
The left summand is identical to the membership queries $\MQ$ in normal DFA learning.
In the equivalence query $\EQ(H,r,q')$, the learner only provides a single successor state $q'$,
describing how to proceed if the hypothesis $H$ was wrong:
\[
  \sem{\EQ(H,r,q')} = \EQ(H,r,(\set{M\in \DFA\mid \exists w\in A^*\colon L(H)(w)\neq L(M)(w)},\,q'))
  \tag*{\agdarefcustom{Restricted EQ Semantics}{DFA-Restricted-Learning}{DFASem.⟦F\_⟧}{}}
\]
Likewise, the generic notion of teacher instantiates to two maps of the form
$T\times A^*\to 2\times T$ and $\DFA\times T\to 1 + T$ (\agdarefcustom{Restricted DFA Teacher}{DFA-Restricted-Learning}{DFATeacher}{}). We can prove formally that there is no learner capable of learning
DFAs in this restricted setting within a polynomial number of queries:
\begin{proposition}[\agdaref{DFA-Examples.DFA-Restricted-Bound}{thm-adversarial-teacher}{}, also sketched in {\cite[Sec.\
3.2]{Angluin87ML}}]\label{thmNoPolyLearn}
  For every input alphabet $k:=|A|$ of size at least $k\ge 2$ and every polynomial $P\colon \N\to \N$, there is a teacher such that
  for every learner some DFA $(Q,q_0,o,\delta)$ is still possible after
  $P(|Q|)$ queries by the learner.
\end{proposition}
\begin{proof}[Proofsketch]
  Given a polynomial $P\colon \N\to\N$, we first construct an exponent $n$ such that:
  \[
    P(2 + n) < k ^ n.
  \]
  This $n$ will be used to construct a word $w\in A^n$ in the end for which
  $L= \set{w}$ will be the witnessing language that the learner has failed to
  learned within $P(2+n)$ queries.
  Here, the query count passes $2+n$ to $P$ because the language $L=\set{w}$ is
  accepted by a DFA with $2+n$ states.

  During learning, the adversarial teacher keeps track of a subset of words
  $T\subseteq A^n$, initially $T=\emptyset$.
  \begin{enumerate}
  \item On a membership query for $v\in A^*$, the teacher replies \textqt{no}
    and sets $T:= T\cup \set{v}$ if $|v| = n$ (otherwise, keep $T$ unchanged).
  \item On an equivalence query for an automaton $H \in\DFA$, the teacher
    checks whether there is some word $v\in A^n$ accepted by $H$. If so,
      put $T :=T\cup \set{v}$ (otherwise, keep $T$ unchanged).
  \end{enumerate}
  After $P(2+n)$ many queries, there is by $P(2 + n) < k ^ n$ some word
  \[
    w \in A^n \setminus T
  \]
  The language $L=\set{w}$ is consistent with all the previous answers by the
  teacher and so not learned in-time by the learner.
\end{proof}

\subsection{Mealy Machines -- Version 1}
Learning algorithms often work with Mealy machines
instead of DFAs (e.g.~$\lsharp$ \cite{VaandragerGRW22}).
A Mealy machine for the input alphabet $A$ and output alphabet $O$ consists of
the data:
\[
  \text{a finite set }Q
  \qquad
  q_0 \in Q
  \qquad
  \delta\colon Q\times A\to O\times Q.
  \tag*{\agdarefcustom{Definition of Mealy machine}{Mealy}{Mealy-Machine}{}}
\]
Here, we can consider different ways to define the teacher's answer to
a membership query (also called \emph{output query}) for an input word $w\in
A^*$:
\begin{enumerate}
\item Version 1: The teacher returns the \emph{last} output symbol to the learner.

  This is analogous to the membership query in DFAs where the teacher only
  responds with the
  acceptance conditions of the last state reached via the input word $w\in
  A^*$, without reporting the acceptance of the intermediate states.

\item Version 2: The teacher returns the outputs of all transitions taken, so
  the teacher returns an output word $v \in O^*$ of the same length as the
  input $w\in A^*$.

\end{enumerate}
The $\lsharp$ learning algorithm works with Version 2 and takes $\CO(k\cdot n^2 +
n\cdot \log(m))$ queries to learn a Mealy machine with $n$ states and
counterexamples of maximum length $m$~\cite[Theorem 3.14]{VaandragerGRW22}.

Does the choice of membership query affect the bound on queries? One may wonder
whether $\lsharp$ saved any queries by using the richer version of membership query.
In order to find out, we first study Version 1 of the membership query, where
only the last output symbol is returned.

The game type is:
\[
  F(R,X) = A^+\times X^{O} + \Mealy \times R \times X^{A^+}
  \tag*{\agdarefcustom{Game type for Mealy machines}{Mealy-Last-CE-Size-Learning}{F}{}}
\]
So the teacher notion instantiates to the following:
\[
  \MQ\colon T \times A^+\to O \times T
  \qquad
  \EQ\colon T\times \Mealy\to 1 + A^+\times T
  \tag*{\agdarefcustom{Teacher for Mealy machines}{Mealy-Last-CE-Size-Learning}{MealyTeacher}{Conversion functions to generic teacher are at the end of the file.}}
\]
Its semantics is defined analogously to DFAs with the bound on counterexample size (\autoref{secCeSize}).

For this stricter setting, we implement the $\lsharp$ algorithm of
Vaandrager et al.~\cite{VaandragerGRW22} in Agda and formally prove its query
bound. We first recall the algorithm and then explain how it instantiates
the proof principle of \autoref{thm:learner-correct}.

\paragraph{The $\lsharp$ algorithm}
The central data structure of $\lsharp$ is an observation tree, which
represents a partial map $\mathcal{T}\colon A^+\partialto O$ collecting the results of
all membership queries so far. $\lsharp$ maintains a prefix-closed set
$S\subseteq A^*$ of words, called the
\emph{basis}, which are (the paths to) the states in the hidden Mealy machine
that have been identified as being distinct.
Initially, $S = \set{\varepsilon}$ is only the empty word and the basis grows whenever
$\lsharp$ discovers a new state whose behaviour is distinct from any of the
existing basis states.
Thus, any two basis states $b_1,b_2\in S$ are \emph{apart}, meaning that some $w\in A^+$ has been
observed for which $b_1\cdot w$ and $b_2\cdot w$ produce different outputs.
Apartness is also how frontier words
are classified: for a frontier word $u\in S\cdot A$, the learner keeps track of
the \emph{candidate} basis words that are not (yet) apart from $u$
(\agdaref{L-Sharp.State}{Semantic-State}{The candidates of the frontier word
$u = b\cdot i$ are the list $\delta\,b\,i$.}).
The learner then acts by the first applicable rule:
\begin{itemize}
\item[(R1)] If some frontier word has \emph{no} candidates, it witnesses a new
  state: promote it to the basis
  (\agdarefcustom{Rule R1}{L-Sharp.SaturateBasis}{saturate-basis}{}).
  This is the only rule that requires no query.
\item[(R2)] If the output of some frontier word $u$ has not been observed
  yet, pose the output query $u$.
\item[(R3)] If some frontier word $u$ still has two distinct candidates
  $b_1,b_2\in S$, query $u\cdot \eta$, where $\eta$ is the suffix witnessing
  the apartness of $b_1$ and $b_2$. By weak co-transitivity of
  apartness~\cite{VaandragerGRW22}, the answer makes $u$ apart from $b_1$ or
  from $b_2$, so at least one candidate is eliminated.
\item[(R4)] Otherwise every frontier word has a unique candidate, and the
  learner builds the evident hypothesis $\mathcal{H}$ on the state set $S$
  and poses it as an equivalence query
  (\agdarefcustom{Rule R4}{L-Sharp.Hypothesis}{hypothesis}{}).
\end{itemize}
In order to find the first matching rule, the we have implemented a do-notation for decision procedures
that already take care of the duality between $\forall$ and $\exists$:
\begin{itemize}
  \item[(R1)] \texttt{IsolatedState?} (\agdarefcustom{Guard R1}{L-Sharp.State}{Semantic-State.IsolatedState?}{})
  decides between
  $\exists b \in S, i \in A\colon \delta(b,i) = \emptyset$
  and $\forall b\in S, i \in A\colon \exists b' \in \delta(b,i)$.
  Here, $\delta(b,i)$ is the list of base states that are not (yet) apart from
  the frontier state $b\cdot i$.

  \item[(R2)] \texttt{LackingFrontier?} (\agdarefcustom{Guard R2}{L-Sharp.State}{Semantic-State.LackingFrontier?}{}):
  decides between
  $\exists b, i\colon \mathcal{T}(b\dot i)\text{ is undefined}$
  and $\forall b, i\colon \mathcal{T}(b\dot i)\text{ is defined}$

  \item[(R3)]
  \texttt{AmbiguousFrontier?} (\agdarefcustom{Guard R3}{L-Sharp.State}{Semantic-State.AmbiguousFrontier?}{}):
  decides between
  $\exists b, i, b_1, b_2: b_1,b_2\in \delta(b,i), b_1\neq b_2$
  and 
  $\forall b, i : |\delta(b,i)| \le 1$

  \item[(R4)] The right-hand disjuncts of (R1), (R2), (R3) yield all the
  sufficient conditions to construct a hypothesis.
\end{itemize}

A counterexample $\rho\in A^+$ to $\mathcal{H}$ is processed by binary
search, in the style of Rivest and Schapire~\cite{RivestS89,RivestS93}: the
learner maintains a word $b\cdot i\cdot\sigma$ (with $b\in S$, $i\in A$)
that is apart from the basis word the hypothesis assigns to it, and each
output query halves the length of the suffix $\sigma$
(\agdarefcustom{Counterexample processing}{L-Sharp.ProcessCE}{Proc-Step}{}).
When $\sigma$ is exhausted, the conflict has been pushed to the frontier
word $b\cdot i$ itself: its last remaining candidate is eliminated, so rule
(R1) fires and the basis grows.

\paragraph{$\lsharp$ as a coalgebra}
In our framework, the learner is a coalgebra $c\colon C\to F(R,C)$ for the
Version~1 game type. The state set
$C$~(\agdarefcustom{State set $C$}{L-Sharp.State}{L\#-State}{In the formalization, the input
alphabet $A$ is named \texttt{I}.}) consists of the observation
tree, basis, and bookkeeping for the phase the learner is in: the main loop
(R1--R4), the two counterexample-processing phases, and a designated
\emph{stuck} state that the learner enters when a teacher response is
inconsistent with \emph{every} Mealy machine -- the game type demands a
continuation for every response, even impossible ones. As the result type we
simply pick $R := \Mealy$. The transition structure
$c$~(\agdarefcustom{Coalgebra $c$}{L-Sharp.Algorithm}{L\#-δ}{}) then selects the first applicable rule.

\begin{theorem}[\agdaref{Main-Results}{L\#.L\#-time-bound%
              }{}]\label{thmLsharp}
  If membership queries return only the last output symbol (Version 1),
  then for $k = |A|$ input symbols,
  the $\lsharp$ learning algorithm learns a Mealy machine with $n$ states
  provided with counterexamples of length at most $m$ by the teacher within
  \[
    (k + 1) \cdot n \cdot (n+1)
    + (n + 1) \cdot \lceil \log_2 m \rceil + 1
  \]
  queries.
\end{theorem}
\begin{proof}[Proof sketch]
As for DFAs (\autoref{secCeSize}),
we take $\D := \Mealy\times \N$, pairing the hidden machine with a bound on the
length of the teacher's counterexamples.
When writing $|M|\in\N$ for the number of states of a Mealy machine $M\in\Mealy$,
we define the bound function by:
\[
  b\colon \D\to \N
  \qquad
  b(M,m) =
    (k + 1) \cdot |M| \cdot (|M|+1)
    + (|M| + 1) \cdot \lceil \log_2 m \rceil + 1
\]
To apply the proof principle \autoref{thm:learner-correct}, we equip the learner with
$\allows$ and $\tick$ (\autoref{defStepBounded}):
\begin{itemize}
\item $\allows(q)\subseteq \D$ contains those $(M,m')$ where $M$ agrees
  with every output recorded in the observation tree of $q$, and $m'$ is at
  least the length of every counterexample received so far
  (\agdarefcustom{\textbullet\ $\allows$ for $\lsharp$}{L-Sharp.State}{L\#-allows}{}). In the
  counterexample-processing phases, $M$ must moreover refute the pending
  hypothesis on the received counterexample.
\item $\tick(q)$ is not a stored step counter but a structural norm of
  the state (\agdarefcustom{\textbullet\ $\tick$ for $\lsharp$}{L-Sharp.State}{Semantic-State.ticks}{}): for $|S| = n_q$,
  it is the sum of
  \[
    \qquad
    \underbrace{\textstyle\sum_{j\le n_q} j}_{\text{(a) basis growth}}
    +
    \underbrace{\#\text{observed frontier outputs}}_{\text{(b) rule R2}}
    +
    \underbrace{\#\text{eliminated candidates}}_{\text{(c) rule R3}}
    +
    \underbrace{n_q\cdot\ceil{1+\log_2 m_q}}_{\text{(d) counterexamples}}
  \]
  where $m_q$ is the length of the longest counterexample received so far.
  Intuitively, the above sum provides the maximum number of queries needed to
  establish all the information we have at the moment:
  \begin{enumerate}
  \item[(a)] for each (unordered)
  pair of basis states, we need one observation that proves the basis states apart.
  \item[(b)] for each basis state $b\in S$ and input $i\in A$, we have observed the output of
    $b\cdot i$.
  \item[(c)] for each frontier state $u\in (S\cdot A)$ and apart basis state $b\in S$,
    we have observed the distinguishing output.
  \item[(d)] each basis state may have been the result of counterexample
    processing, which took $\log_2 m_q$ many output queries.
  \end{enumerate}
\end{itemize}
The learner stays in-time because every query increases the $\tick$ value.
In the main loop, the $\tick$ value is the above norm: an
(R2)-answer defines a new frontier output (b), and an (R3)-answer
eliminates a candidate (c).
During counterexample processing, the queries do not increase any
of the summands; instead, the $\tick$ value of the counterexample phases
additionally counts the number of bisection steps.
When counterexample processing finishes, the number of basis states increases,
which makes the above sum increase by at least
$(n_q+1) + \ceil{1+\log_2 m_q}$, namely
$n_q+1$ in summand (a) and $\ceil{1+\log_2 m_q}$ in summand
(d).

It remains to check $\tick(q) < b(M,m')$ for all $(M,m')\in\allows(q)$,
where $b$ is the bound function defined above
(\agdarefcustom{\textbullet\ Def.\ $b(M,m')$}{L-Sharp.Correctness}{bound}{The formalization proves the invariant
for a slightly finer bound function and relaxes it to the closed form stated
in the theorem.}). Since all basis words are pairwise apart,
they reach pairwise distinct states in any machine consistent with the
observations, so $n_q\le |M|$ for every allowed $(M,m')$.
This bounds each of the four summands, e.g.\ (b) and (c) by $k\cdot |M|$ and
$k\cdot {|M|}^2$. For counterexample processing one proves that every machine that
refutes the current hypothesis must have \emph{strictly more} than $n_q$
states.
Hence, the $\lsharp$ learner is step-bounded by $b$ and so
\autoref{thm:learner-correct} yields the desired bound $b$ on queries.
\end{proof}

So the query complexity of $\lsharp$ remains the same for the stricter membership query variant, and moreover, we obtain the exact factors of the big-$\CO$ complexity result \cite[Theorem 3.14]{VaandragerGRW22}.

\subsection{Mealy Machines -- Version 2}
As an alternative, we consider Version 2 of the return type of membership
queries, and moreover make the equivalence queries not only return the
counterexample but also the output sequence for the counterexample:
\[
  F(R,X) = \coprod_{n\in \N} A^n\times X^{O^n} + \Mealy \times R \times X^{(A\times O)^*}
  \tag*{\agdarefcustom{Game type for Mealy machines}{Mealy-Full-Trace-Learning}{F}{}}
\]
Thus, the game type specifies that the membership queries produce output words of matching length.
The semantics of the Mealy game type (\agdarefcustom{Semantics of Mealy
Learning}{MealyLearning}{MealyLearning-Semantics}{}) is defined analogously to
that of DFAs. Instantiating the notion of teacher to $F$ yields that a teacher
on a set $T$ is equivalent to maps
\[
  \MQ_n\colon T \times A^n\to O^n \times T
  \qquad
  \EQ\colon T\times \Mealy\to 1 + (A\times O)^*\times T
  \tag*{\agdarefcustom{Teacher for Mealy machines}{Mealy-Full-Trace-Learning}{MealyTeacher}{Conversion functions to generic teacher are at the end of the file.}}
\]
Unsurprisingly, any learning algorithm for Version 1 can be adapted to Version 2:
\begin{proposition}[\agdaref{Mealy-Full-Trace-Reduction}{translate-time-bound}{}]
  Every learning algorithm for Mealy machines (Version 1) can be translated to
  the Version 2 setting with the same bound on queries.
\end{proposition}
In the proof, the only subtlety is the counterexample: a counterexample $w \in
(A\times O)^*$ to a hypothesis by the learner can have the mismatching output
anywhere within the output word, whereas in the previous Version 1, a
counterexample $w'\in A^+$ must lead directly to the mismatching output.

\begin{corollary}[\agdaref{Main-Results}{L\#-Full-Trace.L\#-full-trace-time-bound%
              }{}]
  If membership queries return the entire output word (Version 2),
  then for $k = |A|$ input symbols,
  the $\lsharp$ learning algorithm learns a Mealy machine with $n$ states
  provided with counterexamples of length at most $m$ by the teacher within
  \[
    (k + 1) \cdot n \cdot (n+1)
    + (n + 1) \cdot \lceil \log_2 m \rceil + 1
  \]
  queries.
\end{corollary}

\section{Formalization in Agda}\label{agdarefsection}

The results of this paper have been fully formalized in Agda 2.8.0 with the
Agda standard library v2.3. The source code has $\ge 8000$ lines of code and
spans 49 files. The the proofs were written by hand
with the exception of $\lsharp$ proof details in \texttt{Mealy-Full-Trace-Reduction}, \texttt{CartesianSum}, \texttt{Hypothesis},
\texttt{Extend}, \texttt{SaturateBasis}, \texttt{ProcessCE},
\texttt{Correctness}, \texttt{Stepping} and parts of \texttt{Algorithm} which
were created by Claude Fable. The main $\lsharp$ specification in
\texttt{L-Sharp.State} was written by hand.

In the course of the present work, we contributed two new
lemmas on $\log_2$ to the Agda standard library ($2^{\floor{\log_2 n}}\le n$
and $n\le 2^{\ceil{\log_2(n)}}$ in \texttt{Utils.agda}).

The respective HTML files and the Agda source code files can be found on
\begin{center}
\nicehref{\onlineHtmlURL index.html}{\onlineHtmlURL index.html}
\end{center}
and are also directly linked below.

Below we list the \textcolor{srcfilenamecolor}{Agda file} containing the referenced result and (if applicable) mention a concrete identifier (hyperlinked) in this file.

\printcoqreferences

\section{Related Work}
Categorical approaches to automata learning have been studied
extensively~\cite{BarloccoEA19,ColcombetPS21,HeerdtS017,UrbatS20,AristoteGPS25}. In these
frameworks, the categorical abstraction concerns the concept class: the
objects to be learned are modelled as categorical objects (e.g.\ coalgebras for a functor), so that one
learning algorithm uniformly instantiates to different kinds of automata
models, e.g.\ DFAs, Mealy machines, or weighted and nominal automata. The
present work is orthogonal: we model the
\emph{learning algorithm itself} as a coalgebra, whose state space is the
learner's current knowledge and whose transition structure interacts with the
teacher. With this shift of perspective, query complexity becomes a property of this coalgebra
(step-boundedness, \autoref{defStepBounded}), amenable to invariant-style
proof principles. The query complexity of active learning has been analysed
before, both as upper bounds for concrete
algorithms~\cite{ANGLUIN198787,RivestS89,VaandragerGRW22} and as lower
bounds~\cite{BalcazarDG97,KrugerGV23}, but with pen-and-paper proofs, in which
quantifier subtleties remain implicit, because one can simply write \textqt{let
$M$ be the teacher's secret object}. In contrast, reasons about objects that
are consistent with the query history without introspecting the teacher's mind.
There are many formalizations of
automata theoretic results in proof assistants, but we are not aware of any
prior machine-checked treatment of active automata learning, let alone of its
low-polynomial query complexity.

\section{Conclusions and Future Work}
We believe that the present definitions promise a starting point for formalizing
run-time results about learning algorithms in general.
As demonstrated,
the compact definitions of game type, learner, and teacher cover many different variants
of active learning. For
natural numbers, we have seen example implementations of learners for which our
main theorem provided us with a run-time bound in a formalized setting
of Agda. Having proven the bound for $L^\#$~\cite{VaandragerGRW22} already, it
remains for future work to apply this technique to other automata
learning algorithms, e.g.\ $L^*$~\cite{ANGLUIN198787}.

Category-minded readers may have noticed that the functor $M$ in the definition of teacher
is a monad $M$ that models statefulness and
exceptions (\textqt{surrendering})~\cite{Moggi91}. Generalizing $M$ to different
monads may lead to further game types:
\begin{itemize}
\item \textbf{I/O monad:} A general monad would allow teacher implementations
to do actual network I/O to communicate with a remote black-box system. So this is an
instance where teacher does not even know the hidden automaton. Thus, a
formalized learning algorithm can then be applied to actual black box systems.

\item \textbf{Probability distribution monad:} For learning probabiistic
  systems, we can instantiate $M$ with the
  probability distribution monad: Then the actions by the learner stay deterministic,
  but the answers from the teacher are probability distributions.
  This turns the learning game into a probabilistic process.

\end{itemize}

Another direction of future work is to extend the query counter from $\N$ to other
well-ordered sets such that we can analyse the number of membership and equivalence queries
separately in the run-time analysis. Such a general counter mechanism would
also allow to analyse an algoirthm's \emph{symbol complexity}, which does not
only take take the number of queries but also their length into account.

\begin{acks}
The author is grateful for fruitful discussions with Joshua
Moerman on the central notions of this paper; for the suggestion
by Florian Frank to study restricted DFA learning bound; for discussions
with David Wegman on the formal proving capabilities of AI agents;
for comments by Jurriaan Rot and Frits Vaandrager.
\end{acks}

\bibliography{refs}

\end{document}